\documentclass[10pt]{article}
\pdfoutput=1

\usepackage{graphicx}
\usepackage{amssymb}
\usepackage{amsthm}

\newtheorem{theorem}{Theorem}[section]

\newtheorem{proposition}[theorem]{Proposition}

\theoremstyle{remark}

\theoremstyle{definition}
\newtheorem{definition}[theorem]{Definition}

\begin{document}

\title{Eigenfunctions of the Edge-Based Laplacian on a Graph}
\author{Richard C. Wilson\footnote{Corresponding author: Richard.Wilson@york.ac.uk},
 Furqan Aziz, Edwin R. Hancock\footnote{Edwin Hancock is supported by the Royal Society under a Wolfson Research Merit Award.}\\
Dept. of Computer Science\\University of York, York, UK}
\maketitle

\begin{abstract}
In this paper, we analyze the eigenfunctions of the edge-based Laplacian on a graph and
the relationship of these functions to random walks on the graph. We commence
by discussing the set of eigenfunctions supported at the vertices, and demonstrate
the relationship of these eigenfunctions to the classical random walk on the
graph. Then, from an analysis of functions supported only on the interior of
edges, we develop a method for explicitly calculating the edge-interior eigenfunctions of
the edge-based Laplacian. This reveals a connection between the edge-based Laplacian
and the adjacency matrix of backtrackless random walk on the graph. The edge-based eigenfunctions therefore correspond to some eigenfunctions of the normalised Hashimoto matrix. 


MSC[2010] 05C50, 05C81
\end{abstract}


\section{Introduction}
The traditional discrete graph Laplacian operator\cite{Chung:spectral} has proved to be a useful
tool in the analysis of graphs and has found application in a number of areas. For example,
the heat kernel, which is derived from the graph Laplacian, has been
used to define graph kernels\cite{KonLaf02,Lafferty:diffusion} in the
machine learning literature. Sun et al\cite{sun} used the heat kernel on the mesh
representing a 3D shape to create a heat kernel signature for describing shape.
In \cite{Aubry}, Aubry et al used wave-like solutions of Schr\"odinger's equation
to construct an alternative shape descriptor, referred to as the wave kernel signature.
The graph Laplacian was used by Coifman and
Lafon\cite{Coifman2006} for dimensionality reduction of data.
There are many other applications graph Laplacian in the literature.

In \cite{Friedman:somegeometric,FriedmanTillich:calculus}, Friedman and Tillich developed a calculus on graphs which
provides strong connections between graph theory and analysis. This approach has
a number of advantages. It allows the application of many results from analysis
directly to the graph domain, and opens up the use of  many new partial differential
equations on graphs. As an example, they define a wave equation\cite{wave} which
has a finite speed of propagation, in contrast to the usual wave equation on a graph.

In the graph calculus of Friedman and Tillich, the graph is given a geometric realization
by associating an interval with each edge of the graph. Functions may therefore
exist both at the vertices and on the interior of edges.
From this starting point
they develop a divergence and, most importantly, a graph Laplacian. 
This type of Laplacian has found application in the physics literature
where the interpretation is as the limiting case of a ``quantum wire''
\cite{quantumwires,rubin,Kuchment2001}.

The
graph Laplacian consists of two parts; namely a vertex-based Laplacian and an
edge-based Laplacian. Friedman and Tillich also demonstrate that for edgewise-linear functions
the edge-based Laplacian is zero and the graph Laplacian reduces to the traditional
discrete graph Laplacian. On the other hand, for functions where the vertex-based
Laplacian is zero, they obtain the edge-based Laplacian only. This results in
a setting which is substantially
different from the traditional approach. In the remainder of this paper, we will concern
ourselves with the edge-based Laplacian.

While Friedman and Tillich find the eigenvalues of the edge-based Laplacian,
and give some of its eigenfunctions explicitly, they do not give a method
for computing the entire eigensystem of the graph. In this paper, we demonstrate the
relationship between the simple eigenfunctions and the random walk on 
the graph. We then give a method for explicit calculation of the remaining
eigenfunctions. This analysis reveals a link between the remaining eigenfunctions and
the backtrackless random walk on the graph, which in turn is linked to
some properties of the quantum walk on a graph\cite{emms} and its Ihara zeta function\cite{quantum}.
Because of this link, the edge-based eigenfunctions correspond naturally to eigenfunctions
of the normalised adjacency matrix of the oriented line graph (the normalised Hashimoto matrix\cite{Hashimoto}).

The remainder of this paper is organized as follows. In Section 2, we briefly
introduce the formalism of Friedman and Tillich. In Section 3, we detail
the calculation of the vertex-supported eigenfunctions of the
edge-based Laplacian and describe their relationship to the random walk.
Section 4 contains our main result; we provide a method for calculating the remaining eigenfunctions,
an approach which is linked to the backtrackless random walk. Together with
the eigenfunctions detailed in Section 3, this completely describes the eigensystem
in terms of finite matrices on the graph. Finally
in Section 5 we give some conclusions and suggest future directions of research.

\section{The Edge-based Laplacian}
In this section, we give brief details of the formalism of Friedman and 
Tillich\cite{FriedmanTillich:calculus}.
Let $G=(V,E)$ be a graph with a geometric realization ${\cal G}$.
The geometric realization is the metric space consisting of the
vertices $V$ and a closed interval of length $l_e$ associated with
each edge $e\in E$.
The graph has a (possibly empty) boundary
set $\partial G\subset G$. We assume that the boundary is separated,
i.e. that each boundary vertex is incident on only one edge.
The graph {\it \r{G}}=\{{\it \r{V}},{\it \r{E}}\} 
excludes boundary vertices and any incident edges {\it \r{G}}$=G\backslash \partial G$.
We associate an edge variable with each edge. Let $e=(u,v)$
be an edge with interval variable $x_e$. The edge variable $x_e$ equals zero where the edge
meets vertex $u$  and equals one at vertex $v$. The start and end vertices are determined
by assigning an arbitrary orientation to each edge.

\begin{definition}[Friedman and Tillich\cite{FriedmanTillich:calculus}]
A {\it vertex measure}, $\cal V$ is a measure supported on the vertices
with ${\cal V}(v)>0$ for all $v\in V$. For our purposes it suffices to take
${\cal V}(v)=1$ for all $v\in V$. An {\it edge measure}, $\cal E$ is a measure supported on the interior of edges. ${\cal E}(v)=0$ for all $v\in V$ and the restriction
to the interior of edges is the Lebesgue measure.
\end{definition}

Let $f$ be a function defined on the graph (on both edges and vertices).
We take $f(u)$ to mean the value of $f$ at vertex $u$ and $f(e,x_e)$ to mean
the value of $f$ at position $x_e$ along edge $e$.
Since we have different volume measures on the edges and vertices, we must take care 
in dealing with integrating factors since they are different on edges and vertices. We use
$d\cal{V}$ for the vertex integration factor and $d\cal{E}$ for the edge integration factor. As a result,
we have a two-part Laplacian:
\begin{equation}
\Delta f=\Delta_V f d{\cal V} +\Delta_E f d{\cal E},
\end{equation}
where $\Delta_V$ is the {\it vertex-based Laplacian} and $\Delta_E$ is the
{\it edge-based Laplacian}. Since graph Laplacians are usually given as
positive definite operators, the edge-based Laplacian is minus the
usual calculus Laplacian
\begin{equation}
\Delta_E f=-\nabla_\textrm{calc} \cdot \nabla f.
\end{equation}

The vertex-based Laplacian turns out to be
\begin{equation}
\Delta_V f=\frac{1}{{\cal V}(v)}\sum_{e\ni v}
{\bf n}_{e,v}\cdot {\nabla f}|_e(v).
\end{equation}
Here ${\bf n}_{e,v}$ is the {\it outward-pointing} unit normal. In other
words, for an edge $(a,b)$ it points from $a$ to $b$ at the vertex $b$, and
from $b$ to $a$ at the vertex $a$.

\begin{definition}
A function is said to be  {\it edge-based} if $\Delta_V f=0$. For edge-based
functions, the Laplacian consists of only the edge-based part, $\Delta f=\Delta_E f d{\cal E}$
\end{definition}

For a function $f$ to be edge-based, the following condition applies.
\begin{equation}
\sum_{e\ni v}
(-1)^{1-x_{e,v}}\nabla f(e,x_{e,v})=0\:\forall v,
\label{edgecondition}
\end{equation}
or in other words, the sum of the outward-pointing gradients must be zero.

\section{Vertex-supported Edge-based Eigenfunctions}
For the remainder of this paper, we assume that the edge lengths
on the graph are equal. Friedman and Tillich\cite{wave} demonstrate the
connection between the eigenfunctions of the edge-based Laplacian and
the eigenvectors of the row-normalised adjacency matrix. This
follows directly from the observation that $\Delta_E$ is essentially
the familiar Laplacian of calculus and therefore admits eigenvectors of
the form $f(e,x_e)=C(e)\cos(\omega x_e+B(e))$ where $\omega$ is the
frequency of the eigenfunction, corresponding to an eigenvalue of $\omega^2$.
The eigenfunction is edge-based and so applying condition (\ref{edgecondition})
to the eigenfunction gives
\begin{equation}
\sum_{e\ni v} \frac{ f(u)-f(v)\cos\omega}{\sin\omega}=0.
\label{functioncondition}
\end{equation}
for any $\{\omega,f\}$ for $\Delta_E$, when $\omega$ is not
a multiple of $\pi$. We call $\omega$ the {\it frequency} of the eigenfunction 
and $\{\omega,f\}$ an {\it eigenpair} while noting that the corresponding 
eigenvalue is $\omega^2$.

\begin{definition}
A {\it principal eigenpair} is an eigenpair of $\Delta_E$
with $0\leq \omega \leq 2\pi$ where $\omega$ is the smallest magnitude frequency 
of a sequence of eigenpairs
of the form $\{\omega+2n\pi, C(e)\cos[(\omega+2n\pi) x_e+B(e)]\}$, $n\in\mathbb{N}$
with the same coefficients  $B(e)$ and $C(e)$.
\end{definition}

Let the vector ${\bf g}\neq 0$ be the restriction of an eigenfunction $f$ to
the vertices, taken in some particular order, i.e. $g_u=f(u)$. All eigenfunctions in the same sequence
as $f$ have a vertex restriction equal to $\bf g$.

The row-normalised adjacency matrix $\tilde{\bf A}$ for interior vertices
$v\in${\it \r{V}} is given by
\begin{equation}
\tilde{A}_{ij}=\frac{A_{ij}}{\sum_j A_{ij}},
\end{equation}
where $\bf A$ is the usual graph adjacency matrix.

\begin{theorem}[Friedman and Tillich\cite{wave}] Let $\cal G$ be the geometric realization of a
graph and let $\tilde{\bf A}$ be its row-normalized adjacency matrix. Each
eigenvalue $\lambda$ of $\tilde{\bf A}$, with $\lambda\notin\{-1,1\}$
corresponds to two principal frequencies of $\cal G$
\begin{equation}
\begin{array}{ccc}
\cos^{-1}\lambda & \textrm{\it and} & 2\pi-\cos^{-1}\lambda,
\end{array}
\end{equation}
each with the same multiplicity as $\lambda$.
The corresponding eigenvector $\bf g$ is the vertex restriction of the 
sequence of eigenfunctions based on the principal eigenvalue $\omega$.
\end{theorem}

Equation (\ref{functioncondition})
allows us to determine the eigenfunctions. For a
principal frequency $\omega$ the principal eigenfunction 
is $f(e,x_e)=C(e)\cos(B(e)+\omega x_e)$ with
\begin{eqnarray}
C(e)^2&=&\frac{g_v^2+g_u^2-2g_u g_v\cos\omega}{\sin^2\omega}, \label{ccond}\\
\tan B(e)&=&\frac{g_v\cos\omega-g_u}{g_v\sin\omega}. \label{bcond}
\end{eqnarray}
There are two solutions to Equations (\ref{ccond}) and (\ref{bcond}) which are $\{C(e),B_0(e)\}$ or $\{-C(e),B_0(e)+\pi\}$ but both give the
same eigenfunction. The sign of $C(e)$ must be chosen correctly to match the phase, i.e.
so that $C(e)\cos(B(e))=g_u$. Since $B$ and $C$ are uniquely determined, this comprises all
the eigenvectors of this form.

Friedman and Tillich give a pair of frequencies as above for each eigenvalue of
the row-normalized adjacency matrix, but we note that the eigenpair with $\omega=2\pi-\cos^{-1}\lambda$
actually corresponds to the same sequence of eigenpairs as those with $\omega=\cos^{-1}\lambda$, but with negative values of $n$. There is therefore
a single sequence for each eigenvalue with
\begin{eqnarray}
\omega&=&\cos^{-1}\lambda, \\
f(e,x_e)&=&C(e)\cos\left[B(e)+(\omega+2\pi n) x_e\right],n\in\mathbb{Z}. 
\end{eqnarray}

The value $\lambda=1$ is always an eigenvalue of $\tilde{\bf A}$.
This corresponds to a principal frequency of $\omega=0$ in $\Delta_E$ and therefore
the corresponding eigenfunction is $f(e,x_e)=C(e)\cos(B(e))$ which
is constant on the vertices. If $\partial G\neq \emptyset$ then
this eigenfunction must be zero everywhere. If $\partial G= \emptyset$ then we
obtain a single principal eigenpair with $\omega=0$ and $f(e,x_e)=C$.

The value $\lambda=-1$ will be an eigenvalue of
 $\tilde{\bf A}$ if {\it \r{G}} is bipartite. If $\partial G\neq \emptyset$ then the eigenfunction must be zero on the vertices. We defer the evaluation of eigenfunctions
 not supported on $V$ for the next section. If $\partial G= \emptyset$
 then a single principal eigenpair exists with $\omega=\pi$ and $f(e,x_e)=C\cos(\pi x_e)$. The eigenfunction alternates in sign between the two partitions of
 the graph.
 
This comprises all principal eigenpairs which are supported on the vertices
($\bf g\neq 0$)\cite{wave}. The eigenfunctions supported on the vertices are
therefore directly determined by the eigensystem of $\tilde{\bf A}$.

\subsection{Random Walks and Line Graphs}

The line graph $\textrm{LG}(G)=(V_l,E_l)$ of a graph $G$ is a graph construction
in which we replace the edges of $G$ with the vertices of $\textrm{LG}(G)$ as follows. Firstly we create the symmetric digraph $\textrm{SDG}(G)$ of $G$ by replacing each undirected edge with a pair
of oriented edges. Each oriented edge of the SDG then becomes a vertex of  $\textrm{LG}(G)$.
These vertices are connected if the head of one oriented edge meets the tail
of another. The reverse pair of oriented edges are connected, i.e. $((u,v),(v,u))$ is
an edge in the LG.
$$V_l=\{(u,v)\in E(SDG)\},$$
$$E_l=\{((u,v),(v,w)),(u,v)\in E(SDG),(v,w)\in E(SDG)\}.$$

The {\it oriented line graph} $\textrm{OLG}(G)=(V_o,E_o)$ is constructed in the same way as the
LG except that reverse pairs of oriented edges are not connected, i.e. $((u,v),(v,u))$
is {\it not} an edge. The vertex and edge sets of $\textrm{OLG}(G)$ are therefore
$$V_o=\{(u,v)\in E(SDG)\},$$
$$E_o=\{((u,v),(v,w)),(u,v)\in E(SDG),(v,w)\in E(SDG):u\neq w\}.$$

\begin{figure}
\begin{tabular}{ccc}
\includegraphics[width=0.3\linewidth]{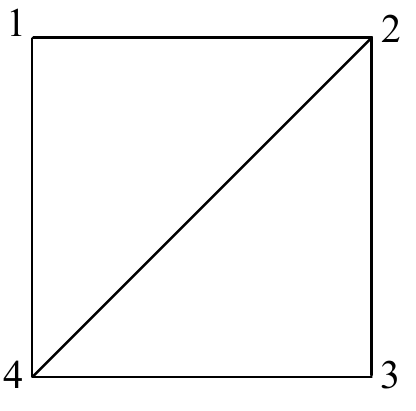} &
\includegraphics[width=0.3\linewidth]{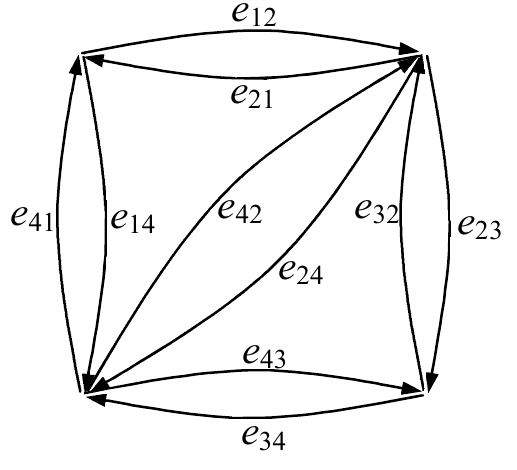} &
\includegraphics[width=0.3\linewidth]{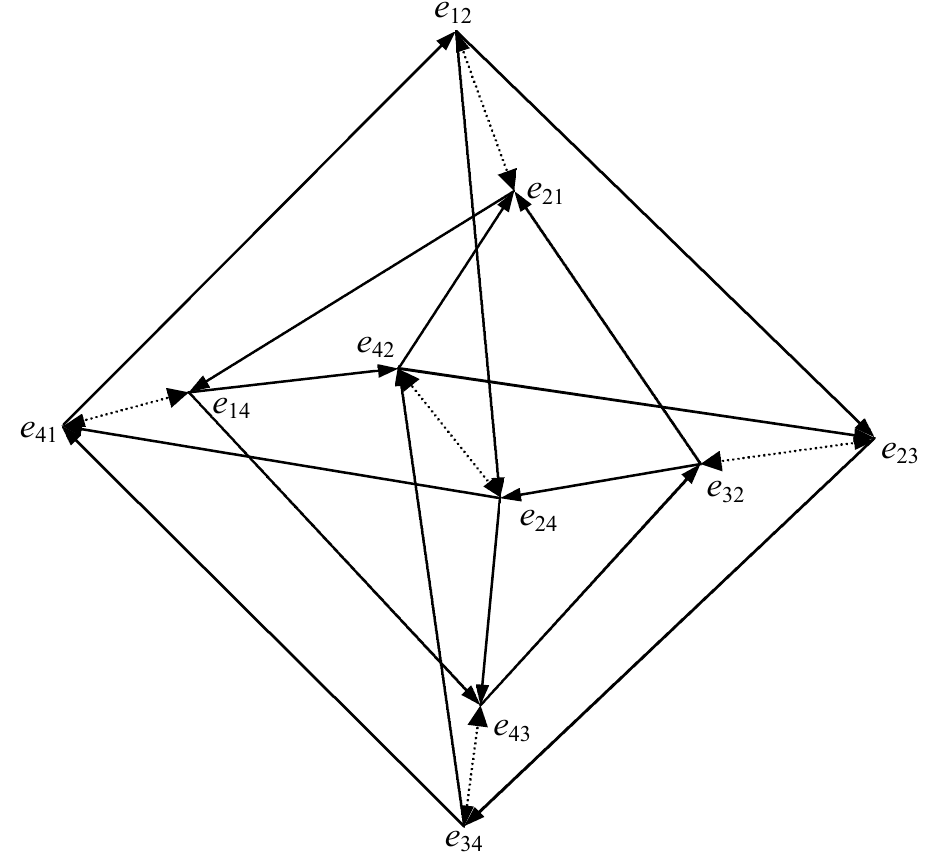}
\end{tabular}
\caption{Left: A graph. Middle: Its symmetric digraph. Right: The line graph and the
oriented line graph; the dotted edges exist in the line graph but not in the oriented line graph.}
\label{fig1}
\end{figure}

Figure \ref{fig1} illustrates these concepts.
A random walk on the vertices of LG represents the sequence of edges traversed in a random walk on the  
original graph $G$. Similarly, a random walk on the OLG represents the sequence of edges in a random walk on $G$ where backtracking steps are not allowed (a {\it backtrackless walk}).

\begin{proposition}
Let $\tilde{\bf A}$ be the row-normalised adjacency matrix of $G$ and $\tilde{\bf U}$ be
the row-normalised adjacency matrix of the line graph of $G$. Each eigenpair $\{\lambda,{\bf g}\}$
of $\tilde{\bf A}$ corresponds to an eigenpair $\{\mu,{\bf h}\}$ of  $\tilde{\bf U}$ with
\begin{eqnarray}
\mu&=&\lambda,\\
h_{uv}&=&A_{uv}g_v.
\end{eqnarray}
\end{proposition}

\begin{proof}
We may write the row-normalized adjacency matrix of the LG as
\begin{equation}
\tilde{\bf U}_{uv,wx}=\frac{A_{uv}A_{wx}\delta_{vw}}{d_x}=A_{uv}\tilde{A}_{wx}\delta_{vw}
\end{equation}
We have
\begin{eqnarray}
\sum_{w,x} \tilde{U}_{uv,wx} h_{wx} & = &
\sum_{w,x} A_{uv}\tilde{A}_{wx}\delta_{vw}A_{wx}g_x, \\
&=&A_{uv} \sum_x \tilde{A}_{vx}g_x, \\
&=&\lambda A_{uv} g_v=\mu h_{uv}.
\end{eqnarray}
\end{proof}

The vertex-supported eigenfunctions of the edge-based Laplacian are therefore determined
by the structure of the random walk on the graph. As we shall show later, the remaining eigenfunctions
are determined by the structure of the backtrackless random walk.

\section{Edge-interior Eigenfunctions}
We now proceed to our main result. All remaining eigenfunctions of the edge-based Laplacian are zero
on the vertices of $\cal G$ and therefore must have a principal frequency
of $\omega\in\{\pi,2\pi\}$. We start with the case $\omega=\pi$.

\begin{proposition}
Principal eigenfunctions of the edge-based Laplacian with principal frequency $\omega=\pi$
and which are zero on the vertices of $\cal G$ are of the form $f(e,x_e)=C(e)\cos(\frac{\pi}{2}+\pi x_e)$, $e\in${\it \r{E}} with
\begin{equation}
\sum_{e\ni v}C(e)=0\:\forall v\in\textrm{\it \r{V}}.
\end{equation}
\end{proposition}

\begin{proof}
Since the boundary is separated, the eigenfunctions must be zero on any
boundary edge and any edge incident on a boundary vertex. As a result, we may
concern ourselves only with {\it \r{G}}. The eigenfunction $f(e,x_e)$ is zero at
both vertices incident on edge $e$, giving values for $B(e)$ of $B(e)\in\{\pi/2,3\pi/2\}$, which both
give the same eigenfunction (with a different sign for $C(e)$). We
may therefore take $B(e)=\pi/2$. The gradients at either end of the
edge are $\nabla f(e,x_e=0)=-\pi C(e)\sin(\frac{\pi}{2})$ and 
$\nabla f(e,x_e=1)=-\pi C(e)\sin(\frac{3\pi}{2})$. Applying
condition (\ref{edgecondition}) to this eigenfunction, we obtain
$$\sum_{e\ni v}C(e)=0.$$
\end{proof}

Hence in order to find eigenfunctions of this type, we must find a
set of coefficients attached to the edges which sum to zero at
every vertex.

\begin{proposition}
The principal eigenfunctions of the edge-based Laplacian with principal frequency $\omega=2\pi$ and which are zero on the vertices of $\cal G$ are of the form $f(e,x_e)=C(e)\cos(\frac{\pi}{2}+2\pi x_e)$, $e\in${\it \r{E}} with
\begin{equation}
\sum_{e\ni v}
(-1)^{1-x_{e,v}}C(e)=0\:\forall v\in\textrm{\it \r{V}}. 
\end{equation}
\end{proposition}

\begin{proof}
The proof is essentially the same as for the previous proposition.
However the gradients at either end of the
edge are $\nabla f(e,x_e=0)=-2\pi C(e)\sin(\frac{\pi}{2})$ and 
$\nabla f(e,x_e=1)=-2\pi C(e)\sin(\frac{5\pi}{2})$. Applying
Condition (\ref{edgecondition}), we obtain
$$\sum_{e\ni v}(-1)^{1-x_{e,v}}C(e)=0.$$
\end{proof}

We may interpret this condition as follows: Consider each undirected edge of {\it \r{G}} as a pair of directed edges. Associate a value $C(e)$ with each directed edge
in the direction of increasing $x_e$ (i.e. from $x_e=0$ to $x_e=1$) and a value $-C(e)$
with the reverse edge. Then the condition above is that the sum of incoming directed
edges at a vertex must be zero.

We may write these conditions for both principal frequencies in the following form. Let $w_{uv}=\pm C(e)$ be the values associated with each edge, as described above. The sign is always positive for $\omega=\pi$ and alternates in sign depending on the direction
of edge traversal for $\omega=2\pi$. Then we may form a matrix $\bf W$ with elements $W_{uv}=w_{uv}A_{uv}$. The conditions for $\omega=\pi$ are then
\begin{eqnarray}
{\bf W}{\bf 1}&=&0,\\
{\bf W}&=&{\bf W}^T,
\end{eqnarray}
and for $\omega=2\pi$ they are
\begin{eqnarray}
{\bf W}{\bf 1}&=&0,\\
{\bf W}&=&-{\bf W}^T,
\end{eqnarray}
where $\bf 1$ is the vector of all-ones.

The oriented line graph of $G$ ($\textrm{OLG}(G)$) represents the structure
of a backtrackless random walk on $G$ in the sense that a random walk
on the vertices of $\textrm{OLG}(G)$ generate a sequence of edges which
can be traversed in a backtrackless random walk on $G$.
In \cite{quantum} we demonstrated that the adjacency matrix
of the oriented line graph ($\bf T$) is equal to the positive support of a quantum
walk on the graph $G$. This matrix is also related to the Ihara zeta function of the graph 
since it is equal to the Perron-Frobenius operator on  $\textrm{OLG}(G)$ and therefore
can be used to calculate the Ihara zeta function $Z_G(u)$ using $Z_G^{-1}(u)=\textrm{det}({\bf I}-u{\bf T})$. We now show that the eigenvectors of $\bf T$ corresponding to eigenvalues of $\lambda=\pm 1$ 
determine the structure of the edge-interior eigenfunctions.

\begin{theorem}
Let $\bf T$ be the adjacency matrix of $\textrm{OLG}(G)$ and $\bf s$ be
an eigenvector of $\bf T$ with eigenvalue $\lambda=1$. Then $s_{uv}=-s_{vu}$ and
$\sum_u s_{uv}=0$, and $\bf s$ provides a solution for ${\bf W}$ in the case of $\omega=2\pi$.
Similarly, if $\lambda=-1$ then  $s_{uv}=s_{vu}$ and
$\sum_u s_{uv}=0$, and $\bf s$ provides a solution for ${\bf W}$ in the case of $\omega=\pi$.
\end{theorem}

\begin{proof}
In \cite{quantum} we demonstrate that $s_{uv}=A_{uv}w_{uv}$ is an
eigenvector of $\bf T$ if $\sum_v A_{uv}w_{uv}=0$ and
either $w_{uv}=-w_{vu}$ or $w_{uv}=w_{vu}$. If $w_{uv}=-w_{vu}$ the
eigenvalue is $\lambda=1$, and if $w_{uv}=w_{vu}$ then the eigenvalue is $\lambda=-1$.
These are precisely the
solutions for $\bf W$ (and for $C(e)$).
Since the eigenvectors
with $\lambda=\pm 1$ span the space of possible solutions, we obtain 
$|E|-|V|+1$ linearly independent solutions for $\lambda=1$ and $|E|-|V|$ linearly independent solutions for $\lambda=-1$  which are all the available solutions
according to \cite{wave}
\end{proof}

The structure of the eigenfunctions which are not supported on the vertices is therefore determined by the eigenvectors of the backtrackless random walk on the graph $G$.

\section{Conclusions}
We have analyzed and completely determined the eigenfunctions of the edge-based Laplacian and
given explicit forms for all of these eigenfunctions. Our analysis provides
a method of computing the eigenfunctions which are zero on the vertices
from the eigenvectors of the oriented line graph. We demonstrate the
connection between the eigenfunctions and both the classical random walk
and the backtrackless random walk. The eigensystem of edge-based Laplacian 
contains eigenfunctions which are related to both the adjacency matrix of
the line graph of $G$ and the adjacency matrix of the oriented line graph of $G$. 

As noted by Friedman and Tillich\cite{wave}, this approach is closer
to traditional analysis than the usual discrete graph Laplacian.
In particular it allows us to formulate wave equations and
relativistic heat equations which have the more usual properties
associated with these equations on a manifold (for example they
will have a finite speed of propagation). The eigensystem of the edge-based
Laplacian
may be of great use in the study of networks where distance 
and propagation speed are important. The current analysis is
currently limited to the case of uniform edge lengths. Future
work will focus on the case where the edge lengths may vary.





\bibliographystyle{plain}
\bibliography{edge}







\end{document}